\documentclass[journal,10pt]{IEEEtran}
\usepackage{amsmath}
\usepackage{amsthm}   
\usepackage{amssymb}
\usepackage{extarrows}%
\usepackage{mdwmath}
\usepackage{mdwtab}
\usepackage{cite}
\usepackage{graphicx}
\usepackage{subfigure}
\usepackage{epstopdf} 
\usepackage[square, comma, sort&compress, numbers]{natbib}
\usepackage{lipsum}    %
\usepackage{multicol}  %
\usepackage{color} 
\theoremstyle{definition}

\newtheorem{rem}{Remark}
\usepackage{algorithm} 
\usepackage{algorithmic}
\usepackage{eqparbox}
\newcommand{\figref}[1]{Fig.~\ref{#1}}
\newcounter{mytempeqncnt} 
\begin{document}

\title{Secure Transmission and Self-Energy Recycling for Wireless-Powered Relay Systems with Partial Eavesdropper Channel State Information
}
\author{\IEEEauthorblockN{  Jingping~Qiao,~\IEEEmembership{Student~Member,~IEEE,}
                           Haixia~Zhang,~\IEEEmembership{Senior~Member,~IEEE,}
                            Feng Zhao, 
                            and Dongfeng~Yuan,~\IEEEmembership{Senior~Member,~IEEE}
                           }
					
	\thanks{J.~Qiao  (e-mail: {\tt qiaojingping@mail.sdu.edu.cn}), H.~Zhang (e-mail: {\tt haixia.zhang@sdu.edu.cn}), and D.~Yuan  (e-mail: {\tt dfyuan@sdu.edu.cn}) are all with Shandong provincial key laboratory of wireless communication technologies, Shandong University, China. 
	F.~Zhao (e-mail:{\tt zhaofeng@guet.edu.cn}) is with the key laboratory of cognitive radio and information processing,  Guilin University of Electronic Technology, China.
	}
}

\maketitle
\pagestyle{empty}  
\thispagestyle{empty} 
\begin{abstract}
This paper focuses on the secure transmission of wireless-powered relay systems with imperfect eavesdropper channel state information (ECSI).  
For efficient energy transfer and information relaying, a novel two-phase protocol is proposed, 
in which the relay operates in full-duplex (FD) mode to achieve simultaneous wireless power and information transmission. 
Compared with those existing protocols, the proposed design possesses two main advantages: 
1) it fully exploits the available hardware resource (antenna element) of relay and can offer higher secrecy rate; 
2) it enables self-energy recycling (S-ER) at relay, in which the loopback interference (LI) generated by FD operation is harvested and reused for information relaying. 
To maximize the worst-case secrecy rate (WCSR) through jointly designing the source and relay beamformers coupled with the power allocation ratio,
an optimization problem is formulated. 
This formulated problem is proved to be non-convex and the challenge to solve it is how to concurrently solve out the beamformers and the power allocation ratio.  
To cope with this difficulty, an alternative approach is proposed by converting the original  
problem into three subproblems. By solving these subproblems iteratively, the closed form 
solutions of robust beamformers and power allocation ratio for the original problem are achieved. 
Simulations are done and results reveal that the proposed S-ER based secure transmission scheme outperforms the traditional time-switching based relaying (TSR) scheme at a maximum WCSR gain of $80\%$.
Results also demonstrate that the WCSR performance of the scheme reusing idle antennas for information reception is much better than that of schemes exploiting only one receive antenna.  
\end{abstract}

\begin{IEEEkeywords}
Physical layer security, full-duplex relay, energy harvesting, self-energy recycling, imperfect channel state information.
\end{IEEEkeywords}

\section{Introduction}

\IEEEPARstart{T}{he} broadcast nature of wireless medium presents
a formidable challenge in ensuring simultaneously secure and
reliable communications in the presence of adversaries \cite{Liu2017survey}. 
To enhance the transmission security, physical layer security, that achieves secrecy by exploiting the inherent randomness of wireless channels, has been proposed recently and drawn intensive attention \cite{Wyner1975,Chen2012TIFS,Ju2017TCOM,Wang2017TCOM}. 
A particularly interesting physical layer security technology is cooperative relaying, since it can not only combat the effect of the path loss by shortening the access distance, 
but also provide extra spatial degrees of freedom for information transmission by  coordinating with other communication nodes.
Recent studies \cite{LunDong2010,Jing2013APCC,Jiang2017SPL,
Chen2015TIFS,Jing2015,Lee2015CL,Rod2015Magaz} show that
compared with direct transmission schemes, cooperative relaying can bring more than twice the secrecy rate gain. 
This has made physical layer security one of the potential technologies to protect the security of 5G wireless networks \cite{5Gpaper}.

The performance of cooperative systems is based on the amount of power of relays spent on information transmission.  
Sometimes, making relay spend more power on information relay is unrealistic. 
For instance, although the mobile users can be treated as cooperative relays in cellular networks, they may unwilling to participate in cooperative communications because of their starving energy storage. 
To alleviate this, wireless-powered strategy is proposed by making use of radio-frequency (RF) signals emitted from other communication nodes to provide stable energy supply for relay nodes \cite{RuiZhang2013TWC,Zhou2013TC,cai2014}. 
With respect to different roles of relays in secure systems, there are two types of cooperative schemes, i.e.,  
wireless-powered cooperative relaying scheme and wireless-powered cooperative jamming scheme. 
Specifically, with harvested energy from RF signals emitted by the source node or the destination node, cooperative relaying schemes \cite{Hong2015ICC,Chen2016TVT,Wang2016TSP,Kalamkar2017TVT} utilize multi-antenna technologies to steer information beams for information enhancement at desired receivers and information leakage degradation at malicious attackers.
While for wireless-powered cooperative jamming schemes \cite{Liu2016TWC,Liu2017CL,Moon2016GLOBECOM}, idle internal communication nodes act as jammers, which consume the harvested energy to transmit jamming signals confusing malicious eavesdroppers. 
All of the aforementioned work focuses on half-duplex (HD) relaying operation, whereas the spectral efficiency loss 
from HD constraint leads to system secrecy rate degradation.

As an attractive solution to overcome this limitation of wireless-powered secure systems, wireless-powered full-duplex (FD) architecture has drawn increasing attention in secure systems \cite{Bi2016JSTSP,Jing2017TVT,Mobini2017ICC}. 
In \cite{Bi2016JSTSP}, a cooperative jamming scheme was developed for friendly jammers with FD architecture, in which jammer nodes simultaneously harvest energy and transmit jamming signals to degrade eavesdropper channels. 
Moreover, 
a novel time switching (TS)-based FD relaying scheme is proposed to enable the secure communications in wireless-powered cooperative systems \cite{Jing2017TVT,Mobini2017ICC}. 
Relay nodes adopt available energy harvested in the first phase to assist information transmission in a FD way. 
Compared with HD operation, FD capable nodes can offer an average of $33\%-66\%$ secrecy rate gain. However, the substantial secrecy performance of FD relaying is limited
by the introduced loopback interference (LI).
Although many advanced interference-cancellation (IC) technologies \cite{FDradios, Sa2014JSAC} have been proposed to suppress LI, implementing such technologies needs 
 additional energy consumption and expensive hardwares. 
A novel self-energy recycling (S-ER) based relaying scheme \cite{Zeng2015WCL} is proposed and investigated to against the negative effects of LI \cite{Kim2016CCNC,Shafie2016WCL,Wu2017WCL}.  
Particularly, 
secure beamforming \cite{Kim2016CCNC} and antenna assignment \cite{Shafie2016WCL} schemes are exploited for S-ER based relaying systems, in which the LI energy is recycled and reused for information retransmission. 
Additionally, as an extension of \cite{Kim2016CCNC,Shafie2016WCL},
\cite{Wu2017WCL} introduced S-ER strategy into FD source nodes, and proposed an optimal beamforming scheme for  
efficient energy and confidential information transmission. 

The above work on S-ER based secure systems can make use of the energy of LI efficiently to achieve better secrecy rate performance. 
However, \cite{Kim2016CCNC,Shafie2016WCL,Wu2017WCL} either focus on single antenna transmission from source to relay or secure systems without cooperative relays, 
leading to limited secrecy enhancement of S-ER relaying schemes. 
More importantly, the potential performance improvements of above work rely on the perfect eavesdropper channel state information (ECSI), which is an ideal assumption \cite{Huang2012TSP}. 
In general, it is prohibitively difficult for a transmitter to obtain full ECSI, since  eavesdroppers are usually keeping silent to hide.  
Moreover, even eavesdroppers feedback channel estimation results to the transmitter, due to some practical issues such as channel estimation errors, feedback errors, quantization and delays, the feedback information may be untrustworthy and exists some uncertainties \cite{Wang2009TSP,Chen2017survey}. Such kind of uncertainties in turn lead to serious information leakage and secrecy performance degradation.

Motivated by the aforementioned problems, this work focuses on 
robust secure transmission scheme design for S-ER based wireless-powered relay systems. 
To facilitate S-ER operation at relay, a two-phase FD relaying protocol is adopted, and separate antenna configuration is considered at the relay node. 
In particular,
half of the time block is devoted for information transmission from source to relay, while the remaining time is used for concurrent information forwarding and energy harvesting. 
The harvested energy from the LI channel in addition to RF signals emitted from the source can be stored and then used for information retransmission in the next time block. 
The main objective of this work is to investigate the effect of imperfect ECSI on the system secrecy rate, and design robust beamformers as well as power allocation scheme 
to degrade the information leakage caused by channel uncertainty, and finally offer more secrecy rate gain. 
Specifically, the main contributions of this paper can be summarized as follows.
\begin{itemize}
\item 
A novel two-phase S-ER based secure system is established by adopting FD mode  
to achieve simultaneous information and power transmission. 
An antenna reuse approach is proposed to make use of 
idle transmit antennas for information reception to enhance the information signal strength. 

\item 
The worst-case secrecy rate (WCSR) is introduced as the system secrecy performance metric.
To maximize WCSR, a joint optimization problem including information beamformer, transceiver beamformers and power allocation ratio is formulated and investigated under imperfect ECSI.

\item 
With respect to the information beamformer, the transceiver beamformers and the power allocation ratio, 
the formulated optimization problem of WCSR is non-convex, and it is prohibitively difficult to solve due to high computational complexity. 
To cope with this difficulty and solve it efficiently, an alternative method is proposed by converting the original problem into three subproblems. 
Solving these subproblems iteratively, the closed form solutions of the robust beamformers and the optimal power allocation ratio maximizing WCSR are derived.

\item 
To demonstrate the superiority of the proposed S-ER based secure system in terms of secrecy performance, the WCSR optimization problem for the traditional time switching based relaying (TSR) systems is also formulated and included as a baseline. 
Simulations are done and results show that the proposed S-ER based secure systems outperform those TSR systems greatly. 
\end{itemize}

The rest of this paper is organized as follows. 
Section II describes the system model of the proposed S-ER based secure system and the WCSR optimization problem formulation. 
Section III elaborates the proposed algorithm to find solutions for the robust beamformers and the power allocation ratio to maximize WCSR performance.  
The WCSR optimization problem for traditional TSR systems is introduced in Section IV as a benchmark for performance comparisons. 
Then, Section V simulates and analyzes the performance of the proposed S-ER based secure system. 
Finally, Section VI concludes all the paper.

\emph{Notations}:
The bold upper (lower) letters denote matrices (column vectors), and $(\cdot)^{\dag}$ denotes the conjugate transpose. 
$\mathbf{I}_{M}$ is the identity matrix of size $M\times M$. 
$\mathbb{E}\{\cdot\}$ denotes expectation, and $\|\cdot\|$ indicates the $L_2$ norm of a vector and the $|\cdot|$ means the absolute value.
$[a]^{+}$ denotes $\max(a,0)$. 
And $\mathcal{CN}(0,\sigma^2)$ represents a circularly symmetric complex Gaussian distribution with zero mean and variance
$\sigma^2$. 
$\mathbb{C}^{N\times M}$ represents the $N \times M$ complex vector space.

\section{System Model and Problem Formulation}

\subsection{System Model}
\begin{figure}[h]
\centering
\includegraphics[width=0.45\textwidth]{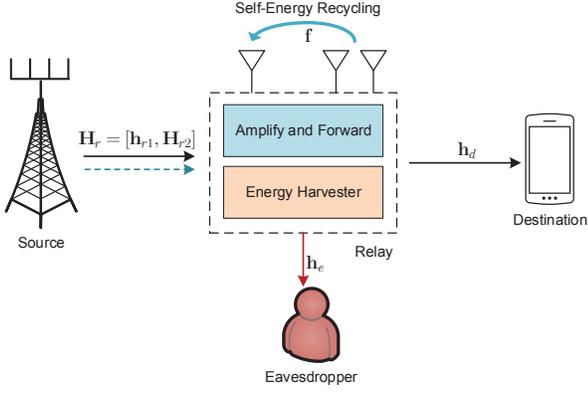}
\caption{System model of the wireless-powered FD relay system, where the S-ER strategy is employed at the FD relay.}
\label{fig_1}
\end{figure}

Secure transmission in a wireless-powered relay system is considered, as shown in \figref{fig_1}. 
In this system, a multi-antenna source transmits confidential signals to a single antenna destination with the help of a 
wireless-powered relay, while keeping information from being intercepted by an 
eavesdropper\footnote{
Note that it is assumed that the direct link between source and destination (or eavesdropper) does not exist due to the deep shadowing effect in the transmission scenario \cite{Hong2015ICC}.
}.
To facilitate simultaneous energy and secure information transmission, a two-phase FD relaying protocol is adopted.   
As shown in \figref{fig_2a}, the source node transmits information signals to the relay node in the first phase, while in the second phase, the relay node works in FD mode to forward information and harvest energy concurrently. 
To enable FD in the second phase, a separate antenna configuration is adopted at relay, i.e., the relay is equipped with one receive antenna and $M$ transmit antennas. 
In the first phase, generally only one receive antenna is used for information reception, while those multiple transmit antennas are idle. 
To make full use of all the available hardware resources (i.e., idle transmit antennas) of relay node to capture more information, an antenna reuse approach at the relay node is proposed. 
The detailed architecture is given in \figref{fig_2b}, where $M$ idle transmit antennas are reused for information reception in the first phase.

\begin{figure}[!t]
\centering
\subfigure[Wireless information and power transmission during one time block $T$.]{
\begin{minipage}{0.45\textwidth}
\centering
\includegraphics[width=0.7\textwidth]{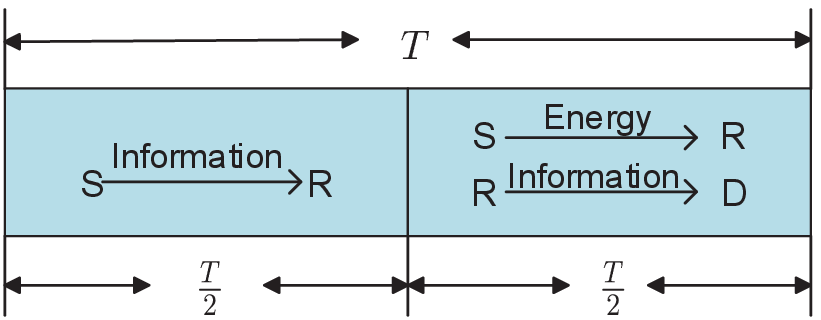}\\
\label{fig_2a}
\end{minipage}
}
\subfigure[Time switching architecture for antenna reuse approach at the relay.]{
\begin{minipage}{0.45\textwidth}
\centering
\includegraphics[width=0.8\textwidth]{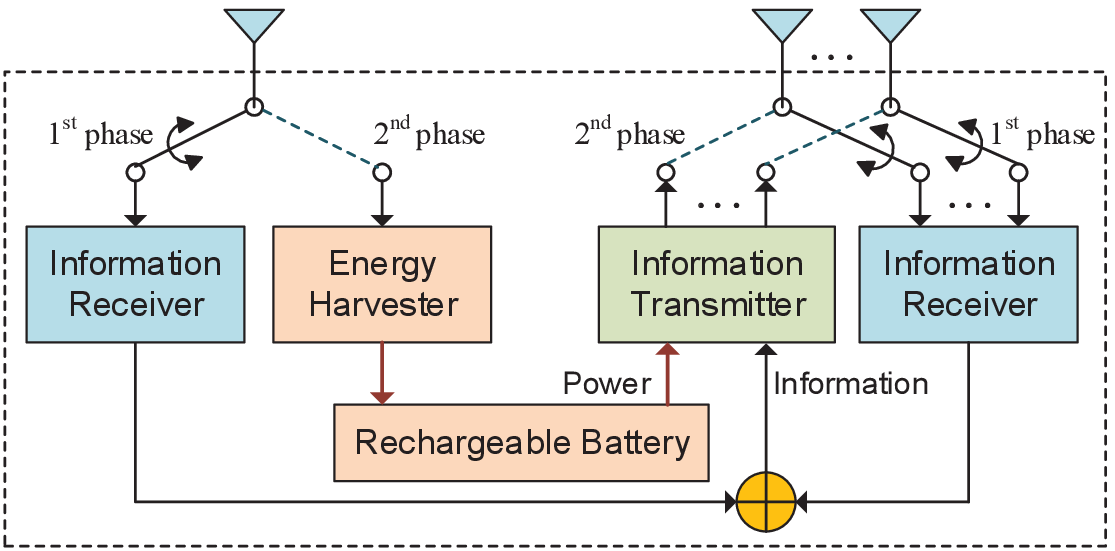}
\label{fig_2b}
\end{minipage}
}
\caption{SWIPT model and 
time switching architecture for antenna reuse approach at the wireless-powered relay.}
\label{fig_2}
\end{figure}

With the antenna reuse approach, the received signal at the relay node during the first phase is written as 
\begin{equation}\label{eq_1}
\mathbf{y}_{r,1}=\mathbf{H}_{r}^{\dag}\mathbf{w}_{s}x_{s}+\mathbf{n}_{r},
\end{equation}
where $\mathbf{y}_{r,1} \in \mathbb{C}^{(M+1)\times 1}$,  $\mathbf{w}_{s}x_s$ denotes confidential information signal transmitted from source node with 
$\mathbb{E}\{|x_{s}|^2\}=P_s$, and $\mathbf{w}_s$ represents the information beamforming vector at the source node with unit power
 $\mathbb{E}\{\|\mathbf{w}_s\|^2\}=1$. 
The channel from source to relay $\mathbf{H}_{r}=[\mathbf{h}_{r1}, \mathbf{H}_{r2}]$ is assumed to be perfectly known to legitimate nodes, 
where $\mathbf{h}_{r1} \in \mathbb{C}^{N\times 1}$ and $\mathbf{H}_{r2} \in \mathbb{C}^{N\times M}$ represent channel coefficients from source to receive antenna  and those transmit antennas at relay, respectively.
Moreover, $\mathbf{n}_{r}$ denotes the additive white Gaussian noise (AWGN) at relay, i.e., $\mathbf{n}_{r}\sim \mathcal{CN}(0,N_{0}\mathbf{I}_{M+1})$.

In the second phase, to enable simultaneous wireless information and power transfer, the relay node works in FD mode. 
Different from traditional FD relaying protocols, the LI caused by FD operation is beneficial, 
since it is also a power source for the energy-constrained relay and 
enables the relay to reuse its own transmitted energy. 
This mode is the previously described S-ER \cite{Zeng2015WCL,Shafie2016WCL}. 
Therefore, adopting S-ER protocol at relay, no extra energy and hardware resources needed any more to suppress the LI. 
Moreover, the energy of LI can be harvested together with the energy signal emitted from the source, i.e., 
\begin{equation}\label{eq_2}
y_{r,2}=\mathbf{h}_{r1}^{\dag}\mathbf{w}_{H}x_{p}+
\mathbf{f}^{\dag}\mathbf{w}_{t}x_{r}+n_{r,H}.
\end{equation}
On the right-hand side of equation (\ref{eq_2}), the first part is the energy signal from the source node, where $\mathbf{w}_{H}$ denotes the energy beamformer with $\mathbb{E}\{\|\mathbf{w}_H\|^2\}=1$, and $x_{p}$ is the energy signal emitted from the source with power $\mathbb{E}\{|x_{p}|^2\}=P_s$.  
Additionally, the second part represents the LI signal, wherein $\mathbf{f}$ denotes the LI channel between antennas at FD relay and is normally regarded as Rayleigh fading channel \cite{Chen2015TIFS,Ikri2012TWC}.
While $\mathbf{w}_{t}x_{r}$ represents the retransmit signal with $\mathbb{E}\{|x_r|^2\}=P_r$, and $\mathbf{w}_{t} \in \mathbb{C}^{M\times 1}$ is the transmit beamforming vector at the relay with $\mathbb{E}\{\|\mathbf{w}_t\|^2\}=1$. 
Besides, as the energy transmission between source and relay can be seen as a MISO (Multi-Input Single-Output) system, MRT (Maximum Ratio Transmission) is employed as the energy beamformer, i.e., $\mathbf{w}_{H}=\frac{\mathbf{h}_{r1}}{\|\mathbf{h}_{r1}\|}$. 
Then, the harvested energy is written as
\begin{equation}\label{eq_3}
E_{H}=\frac{T}{2}\eta (|\mathbf{h}_{r1}^{\dag}\mathbf{w}_{H}|^2P_{s}+|\mathbf{f}^{\dag}\mathbf{w}_{t}|^2P_r).
\end{equation}
where $\eta$ is the energy harvesting efficiency, $0 < \eta\leq 1$, which depends on the rectification process and the energy harvesting circuitry \cite{Zhou2013TC}.
Here, the power of noise (including both the antenna noise and the rectifier noise) is not taken into consideration, since it  is much smaller than that of energy signal \cite{Mou2015WCSP}.

The harvested energy at relay will be stored in a battery for information relaying in the next transmission block.
Considering the long-term time variations, the needed transmit power at relay is not higher than the harvested power, i.e., $P_r\leq 
\eta (|\mathbf{h}_{r1}^{\dag}\mathbf{w}_{H}|^2P_{s}+|\mathbf{f}^{\dag}\mathbf{w}_{t}|^2P_r)$. 
$\delta$ is defined as the power allocation ratio, $0 \leq \delta \leq 1$, which denotes $\delta$ part of the harvested energy is allocated for information transmission. Then the analytical expression of $P_r$ can be expressed as
\begin{equation}\label{eq_4}
P_{r}=\frac{\delta\eta P_{s}|\mathbf{h}_{r1}^{\dag}\mathbf{w}_{H}|^2}{1-\delta\eta|\mathbf{f}^{\dag}\mathbf{w}_{t}|^2}. 
\end{equation}

Assume that a receive beamformer vector, $\mathbf{w}_{r}\in \mathbb{C}^{(M+1)\times 1}$, is employed at relay, the transmit signal will be $\mathbf{w}_{r}^{\dag}\mathbf{y}_{r,1}$.  
Employing amplify-and-forward (AF) protocol, the transmit signal is first amplified by the amplify factor $\beta$ and processed by the transmit beamformer $\mathbf{w}_{t}\in \mathbb{C}^{M\times1}$, and then forwarded to the destination.  
Therefore, the processed transmit signal in this phase can be denoted as $\mathbf{w}_{t}x_r$, where 
 $x_r= \beta \mathbf{w}_{r}^{\dag}\mathbf{y}_{r,1}$, and the amplify factor is expressed as $\beta= \sqrt{P_r/(P_s|\mathbf{w}_{r}^{\dag}\mathbf{H}_{r}^{\dag}\mathbf{w}_{s}|^2+N_0)}$.

Finally, the received signal at the destination can be written as
\begin{eqnarray}\label{eq_5}
\begin{split}
y_{d}=&\mathbf{h}_{d}^{\dag}\mathbf{w}_{t}x_{r}+n_{d}\\
=&\mathbf{h}_{d}^{\dag}\mathbf{w}_{t}\left(\beta\mathbf{w}_{r}^{\dag}\mathbf{H}_{r}^{\dag}\mathbf{w}_{s}x_{s}
+\beta\mathbf{w}_{r}^{\dag}\mathbf{n}_{r}\right)+n_d,\\
\end{split}
\end{eqnarray}
where $\mathbf{h}_d$ is the relay-to-destination channel vector, which is assumed to be perfectly known to legitimate nodes, and $n_d \sim \mathcal{CN}(0,N_{0})$ denotes the noise at destination. Meanwhile, the received signal at the eavesdropper is expressed as 
\begin{eqnarray}\label{eq_6}
\begin{split}
y_{e}
=&\mathbf{h}_{e}^{\dag}\mathbf{w}_{t}\left(\beta\mathbf{w}_{r}^{\dag}\mathbf{H}_{r}^{\dag}\mathbf{w}_{s}x_{s}
+\beta\mathbf{w}_{r}^{\dag}\mathbf{n}_{r}\right)+n_e,\\
\end{split}
\end{eqnarray}
where $n_e \sim \mathcal{CN}(0,N_{0})$ represents the noise at the eavesdropper, and $\mathbf{h}_e$ denotes the relay-to-eavesdropper channel coefficient, which is perfectly known to the eavesdropper but imperfectly known to those legitimate nodes \cite{Khi2010TIT}. 
Without loss of generality, the imperfect ECSI is modeled as additive uncertainty model $\mathbf{h}_e=\bar{\mathbf{h}}_e+\Delta\mathbf{h}_e$, where $\bar{\mathbf{h}}_e$ represents the estimated channel vector of the eavesdropper, $\Delta\mathbf{h}_e$ represents the channel uncertainty. 
It is assumed that the uncertainty is deterministically bounded \cite{Wang2013TVT}, i.e., $\|\Delta \mathbf{h}_e\|\leq \epsilon$.

\subsection{Problem Formulation}
In order to guarantee secrecy for any admissible ECSI uncertainties, including worst, the worst-case secrecy rate (WCSR) \cite{Wang2013TVT,Li2012TSP,Huang2012TSP} is adopted as the performance metric, which is defined as the minimum secrecy rate for any error in the uncertainty region of the channel, 
\begin{eqnarray}\label{eq_7}
\begin{split}
R_{wc}=&\min_{\|\Delta \mathbf{h}_e\| \leq \epsilon} \frac{1}{2}[\log_{2}(1+\gamma_d)-\log_2(1+\gamma_{e})]^{+}\\
=&\frac{1}{2}[\log_{2}(1+\gamma_d)-\log_2(1+\gamma_{ewc})]^{+},
\end{split}
\end{eqnarray}
where $\gamma_d$ represents the received SINR at the destination and is expressed as
\begin{eqnarray}\label{eq_8}
\begin{split}
\gamma_{d}=& \frac{P_s\beta^2|\mathbf{h}_{d}^{\dag}\mathbf{w}_{t}|^2|\mathbf{w}_{r}^{\dag}\mathbf{H}_{r}^{\dag}\mathbf{w}_{s}|^2}{\beta^2 |\mathbf{h}_{d}^{\dag}\mathbf{w}_{t}|^2\|\mathbf{w}_{r}^{\dag}\|^2N_0+N_0}\\
=&\frac{P_s P_r|\mathbf{h}_{d}^{\dag}\mathbf{w}_{t}|^2|\mathbf{w}_{r}^{\dag}\mathbf{H}_{r}^{\dag}\mathbf{w}_{s}|^2}{P_rN_0|\mathbf{h}_{d}^{\dag}\mathbf{w}_{t}|^2+N_0P_s|\mathbf{w}_{r}^{\dag}\mathbf{H}_{r}^{\dag}\mathbf{w}_{s}|^2+N_0^2},\\
\end{split}
\end{eqnarray}
while $\gamma_{ewc}$ denotes the maximum SINR the eavesdropper can get under any channel realization. 
To obtain the expression of $\gamma_{ewc}$, the following inequalities associated with imperfect wiretap channel are first given.

According to triangle inequality and the bound of the channel uncertainties, the following inequality holds
\begin{equation}\label{eq_9}
|(\bar{\mathbf{h}}_e+\Delta \mathbf{h}_e)^{\dag}\mathbf{w}_{t}|
\leq |\bar{\mathbf{h}}_e^{\dag}\mathbf{w}_{t}|+|\Delta \mathbf{h}_e^{\dag}\mathbf{w}_{t}|. 
\end{equation}
Besides, applying the Cauchy-Schwarz inequality to the second term in the right-hand-side of (\ref{eq_9}), we have 
\begin{equation}\label{eq_10}
|\Delta \mathbf{h}_e^{\dag}\mathbf{w}_{t}|\leq \|\Delta \mathbf{h}_e\| \cdot \|\mathbf{w}_{t}\| \leq \epsilon \|\mathbf{w}_{t}\| . 
\end{equation}
Plugging (\ref{eq_9}) and (\ref{eq_10}), the following inequality is achieved
\begin{equation}\label{eq_11}
|(\bar{\mathbf{h}}_e+\Delta \mathbf{h}_e)^{\dag}\mathbf{w}_{t}|
\leq |\bar{\mathbf{h}}_e^{\dag}\mathbf{w}_{t}|+\epsilon \|\mathbf{w}_{t}\|.
\end{equation}
Then, $\gamma_{ewc}$ can be expressed as
\begin{eqnarray}\label{eq_12}
\begin{small}
\begin{split}
\gamma_{ewc}=&\max_{\|\Delta \mathbf{h}_e\| \leq \epsilon}\frac{P_s P_r|\mathbf{h}_{e}^{\dag}\mathbf{w}_{t}|^2|\mathbf{w}_{r}^{\dag}\mathbf{H}_{r}^{\dag}\mathbf{w}_{s}|^2}{P_rN_0|\mathbf{h}_{e}^{\dag}\mathbf{w}_{t}|^2+N_0P_s|\mathbf{w}_{r}^{\dag}\mathbf{H}_{r}^{\dag}\mathbf{w}_{s}|^2\!+\!N_0^2}\\
=&\frac{P_s P_r(|\bar{\mathbf{h}}_{e}^{\dag}\mathbf{w}_{t}|+\epsilon\|\mathbf{w}_{t}\|)^2|\mathbf{w}_{r}^{\dag}\mathbf{H}_{r}^{\dag}\mathbf{w}_{s}|^2}
{P_r N_0(|\bar{\mathbf{h}}_{e}^{\dag}\mathbf{w}_{t}|+\epsilon\|\mathbf{w}_{t}\|)^2\!+\!N_0P_s|\mathbf{w}_{r}^{\dag}\mathbf{H}_{r}^{\dag}\mathbf{w}_{s}|^2\!+\!N_0^2}\\
\end{split}
\end{small}
\end{eqnarray}

The detailed analytical expression of WCSR can be obtained by substituting (\ref{eq_8}) and (\ref{eq_12}) into (\ref{eq_7}). 

Finally, with the objective to maximize WCSR performance through jointly designing the source and relay beamformers as well as the relay power allocation ratio, 
the optimization problem is formulated as
\begin{equation*}
(\mathrm{P}1):~\max_{\delta,\mathbf{w}_{s},\mathbf{w}_{r},\mathbf{w}_{t}} R_{wc}(\delta, \mathbf{w}_{s},\mathbf{w}_{r},\mathbf{w}_{t})
\end{equation*}
\begin{equation}\label{eq_16}
s.t. \left\{\begin{array}{ll} 
&\|\mathbf{w}_{s}\|^2=1\\
&\|\mathbf{w}_{r}\|^2=1\\
& \|\mathbf{w}_{t}\|^2=1\\
&P_{r}=\frac{\delta\eta P_{s}|\mathbf{h}_{r1}^{\dag}\mathbf{w}_{H}|^2}{1-\delta\eta|\mathbf{f}^{\dag}\mathbf{w}_{t}|^2}\\
& 0\leq \delta \leq 1\\
\end{array}\right..
\end{equation}

Since the beamforming vectors $\mathbf{w}_s$, $\mathbf{w}_r$ and $\mathbf{w}_t$ as well as the relay power allocation ratio $\delta$ are coupled with each other, it is obvious that the formulated WCSR optimization problem is non-convex. 
Hence, the complexity of finding its optimal solution is prohibitively high. It is even impossible to solve out these four variables concurrently. As an alternative, to solve it, the original problem $\mathrm{P}1$ is converted into three subproblems, which can be solved iteratively. 

\section{Worst-Case Secrecy Rate Optimization}
As mentioned, the formulated optimization problem $\mathrm{P}1$ is non-convex and difficult to solve. To find out the solutions, it is proposed to convert the optimization problem into three subproblems with respect to $\delta$, $(\mathbf{w}_{s},\mathbf{w}_{r})$ and $\mathbf{w}_{t}$, respectively, based on the fact that  
both information beamformer and receive beamformer are coupled with source-to-relay channel coefficients. 
By solving the subproblem of power allocation ratio, the analytical expression of  $\delta^{\star} (\mathbf{w}_{s},\mathbf{w}_{r},\mathbf{w}_{t})$ can be obtained. Substituting it into WCSR in (\ref{eq_7}), the subproblem of $(\mathbf{w}_{s},\mathbf{w}_{r})$ can also be solved. 
Similarly, $\mathbf{w}_{t}$ can be derived. Finally, by substituting the obtained  beamformers into the expression of $\delta^{\star}$, the power allocation ratio is updated, and $(\delta^{\star}, \mathbf{w}_{s}^{\star},\mathbf{w}_{r}^{\star}, \mathbf{w}_{t}^{\star})$ for the original problem can be obtained accordingly.

Before analyzing the secrecy performance, the positive WCSR constraint is first derived in the following proposition. 
\newtheorem{prop}{Proposition}
\begin{prop}\label{prop_1}
If and only if the received information power at the destination is larger than that at the eavesdropper,  i.e., $|\mathbf{h}_{d}^{\dag}\mathbf{w}_{t}|^2>(|\bar{\mathbf{h}}_{e}^{\dag}\mathbf{w}_{t}|+\epsilon\|\mathbf{w}_{t}\|)^2$, the positive WCSR $R_{wc}>0$ can be achieved.
\end{prop}
\begin{proof}
The details are given in Appendix A.
\end{proof}

\subsection{Power Allocation}
From the analytical expression of relay power in (\ref{eq_4}), it is easy to know that the power is associated with beamforming vectors. 
Therefore, the optimal power allocation scheme is first investigated. The WCSR optimization problem with respect to power allocation ratio $\delta$ is expressed as
\begin{equation*}
(\mathrm{P}2.1):~\max_{\delta} R_{wc}(\delta)
\end{equation*}
\begin{equation}
s.t. \left\{\begin{array}{ll} 
&P_{r}=\frac{\delta\eta P_{s}|\mathbf{h}_{r1}^{\dag}\mathbf{w}_{H}|^2}{1-\delta\eta|\mathbf{f}^{\dag}\mathbf{w}_{t}|^2}\\
& 0\leq \delta \leq 1\\
\end{array}\right..
\end{equation}
In order to get the optimal power allocation ratio, the following Proposition \ref{prop_delta} is proposed.
\begin{prop}\label{prop_delta}
With the positive WCSR constraint, 
the optimal solution of relay power allocation ratio is 
\begin{equation}\label{de_pt}
\delta^{\star}=\left\{
\begin{array}{lcl}
\frac{1}{\eta|\mathbf{f}^{\dag}\mathbf{w}_{t}|^2+Q},
&C_1~or~C_2~or~C_3\\
1, & otherwise\\
\end{array}
\right.
\end{equation}
where 
\begin{equation}
Q=\eta P_s|\mathbf{h}_{r1}^{\dag}\mathbf{w}_{H}|^2\sqrt{\frac{|\mathbf{h}_{d}^{\dag}\mathbf{w}_{t}|^2(|\bar{\mathbf{h}}_{e}^{\dag}\mathbf{w}_{t}|+\epsilon\|\mathbf{w}_{t}\|)^2}
{N_0P_s|\mathbf{w}_{r}^{\dag}\mathbf{H}_{r}^{\dag}\mathbf{w}_{s}|^2+N_0^2}}. 
\end{equation}
The constraints in (15) are denoted as 
$C_1=\{0<Q<1,\eta|\mathbf{f}^{\dag}\mathbf{w}_{t}|^2>\max(Q,1-Q)\}$, $C_2=\{Q>\frac{1}{2},1-Q<\eta|\mathbf{f}^{\dag}\mathbf{w}_{t}|^2<\min(1,Q)\}$, and $C_3=\{\frac{1}{2}<Q<1,\eta|\mathbf{f}^{\dag}\mathbf{w}_{t}|^2=Q\}$, respectively .
\end{prop}
\begin{proof}
See Appendix B.
\end{proof}

After obtaining the optimal power allocation ratio $\delta^{\star}$ through Proposition \ref{prop_delta}, the optimal transmit power at the relay can be written as
\begin{equation}
P_r^{\star}=\left\{
\begin{array}{lcl}
\sqrt{\frac{N_0P_s|\mathbf{w}_{r}^{\dag}\mathbf{H}_{r}^{\dag}\mathbf{w}_{s}|^2+N_0^2}{|\mathbf{h}_{d}^{\dag}\mathbf{w}_{t}|^2(|\bar{\mathbf{h}}_{e}^{\dag}\mathbf{w}_{t}|+\epsilon\|\mathbf{w}_{t}\|)^2}},
&C_1~or~C_2~or~C_3 \\
\frac{\eta P_s|\mathbf{h}_{r1}^{\dag}\mathbf{w}_{H}|^2}{1-\eta|\mathbf{f}^{\dag}\mathbf{w}_{t}|^2}, & otherwise\\
\end{array}\right..
\end{equation}

\subsection{Information and Receive Beamforming Design}

In expressions of $\gamma_d$ and $\gamma_{ewc}$, the source information beamformer and relay receive beamformer are always coupled with the source-to-relay channel  coefficients, i.e.， $|\mathbf{w}_{r}^{\dag}\mathbf{H}_{r}^{\dag}\mathbf{w}_{s}|^2$. 
Based on the obtained optimal power allocation ratio $\delta^{\star}$, these two beamformers can be jointly designed by maximizing WCSR  
\begin{equation*}
(\mathrm{P}2.2):~\max_{\mathbf{w}_{s},\mathbf{w}_{r}} R_{wc}(\mathbf{w}_{s},\mathbf{w}_{r})
\end{equation*}
\begin{equation}
\begin{small}
s.t. \left\{\begin{array}{ll} 
&\|\mathbf{w}_{s}\|^2=1\\
&\|\mathbf{w}_{r}\|^2=1\\
& \delta=\delta^{\star}\\
\end{array}\right..
\end{small}
\end{equation}
For finding out the optimal solution of $\mathrm{P}2.2$, the following proposition is proposed.
\begin{prop}
For any value of power allocation ratio $\delta^{\star}$, only if the maximum value of $|\mathbf{w}_{r}^{\dag}\mathbf{H}_{r}^{\dag}\mathbf{w}_{s}|^2$ is achieved the WCSR performance can be maximized. 
\end{prop}
\begin{proof}
See Appendix C.
\end{proof}

Consequently, the problem $\mathrm{P}2.2$ is equivalent to maximize $|\mathbf{w}_{r}^{\dag}\mathbf{H}_{r}^{\dag}\mathbf{w}_{s}|^2$ under power constraints, i.e.,
\begin{equation*}
(\mathrm{P}2.2.1):~\max_{\mathbf{w}_{s},\mathbf{w}_{r}} |\mathbf{w}_{r}^{\dag} \mathbf{H}_{r}^{\dag}\mathbf{w}_{s}|^2
\end{equation*}
\begin{equation}
s.t. \left\{\begin{array}{ll} 
&\|\mathbf{w}_{s}\|^2=1\\
&\|\mathbf{w}_{r}\|^2=1\\
\end{array}\right..
\end{equation}
Doing singular value decomposition (SVD) of $\mathbf{H}_{r}^{\dag}$, we have 
$\mathbf{H}_{r}^{\dag}=\mathbf{U}\mathbf{\Sigma}\mathbf{V}^{\dag}$, where 
$\mathbf{\Sigma} \in \mathbb{R}^{(M+1)\times N}$  is a rectangular diagonal matrix,  whose diagonal entries $\sigma_{i}$ are known as singular values of $\mathbf{H}_{r}^{\dag}$. 
Both $\mathbf{U} \in \mathbb{C}^{(M+1)\times (M+1)}$  and  $\mathbf{V} \in \mathbb{C}^{N\times N}$ are unitary matrices, the columns of $\mathbf{U}$ and the columns of $\mathbf{V} $ are left-singular vectors and right-singular vectors of $\mathbf{H}_{r}^{\dag}$, respectively.
Based on the SVD properties of matrix $\mathbf{H}_{r}^{\dag}$, the optimal solution of problem $\mathrm{P}2.2.1$ can be obtained
\begin{equation}\label{eq_23}
\left(\mathbf{w}_{s}^{\star}, \mathbf{w}_{r}^{\star}\right)=\left(\mathbf{v}^{unit}, \mathbf{u}^{unit}\right),
\end{equation}
where $\mathbf{v}^{unit}$ and $\mathbf{u}^{unit}$ are 
the right-singular vector and the left-singular vector corresponding to the largest singular value, with $\|\mathbf{u}^{unit}\|^2=1$, $\|\mathbf{v}^{unit}\|^2=1$. 
Since problems $\mathrm{P}2.2$ and $\mathrm{P}2.2.1$ are equivalent, the solution in  (\ref{eq_23}) is also the optimal solution of the original problem $\mathrm{P}2.2$.

\subsection{Relay Transmit Beamforming Design}

After obtaining the optimal information beamformer, the optimal receive beamformer as well as the optimal power allocation ratio, 
the original WCSR optimization problem $\mathrm{P}1$ is simplified to be a function of the transmit beamformer $\mathbf{w}_{t}$. 
To degrade the information leakage and null out the confidential signals at the eavesdropper, zero-forcing (ZF) beamforming is adopted, i.e., 
 $\bar{\mathbf{h}}_{e}^{\dag}\mathbf{w}_{t}=0$. 
Hence, the WCSR optimization problem with respect to $\mathbf{w}_{t}$ is reformulated as 
\begin{equation*}
(\mathrm{P}2.3):~\max_{\mathbf{w}_{t}} R_{wc}(\mathbf{w}_{t})
\end{equation*}
\begin{equation}
\begin{small}
s.t.\left\{\begin{array}{ll}  
& \bar{\mathbf{h}}_{e}^{\dag}\mathbf{w}_{t}=0\\
& \|\mathbf{w}_{t}\|^2=1\\
& (\mathbf{w}_{s}, \mathbf{w}_{r})=
(\mathbf{w}_{s}^{\star}, \mathbf{w}_{r}^{\star})\\
& \delta=\delta^{\star}\\
\end{array}\right..
\end{small}
\end{equation}
To facilitate ZF beamforming design, a matrix $\mathbf{B}$ is defined to accommodate the orthogonal vectors that form the basis of the null space of $\bar{\mathbf{h}}_{e}^{\dag}$, i.e., the right singular vectors corresponding to zero singular values of $\bar{\mathbf{h}}_{e}^{\dag}$. 
Thereby, $\mathbf{w}_{t}$ is designed as a linear combination of the basis of the null space of $\bar{\mathbf{h}}_{e}^{\dag}$, i.e., $\mathbf{w}_{t}^{\star}=\mathbf{B}\mathbf{v}$, where $\mathbf{v}$ is a column vector. 
 As the optimal $\delta^{\star}$ in a piecewise expression,  
the solution of the transmit beamforming vector also has two possibilities, which are analyzed in the following.
\subsubsection{Part of harvested energy for information transmission}

For $\delta^{\star}=1/(\eta|\mathbf{f}^{\dag}\mathbf{w}_{t}|^2+Q)$, the received SINRs at the destination and eavesdropper can be rewritten as 
\begin{equation}
\begin{small}
\gamma_{d}=\frac{P_s|\mathbf{w}_{r}^{\dag}\mathbf{H}_{r}^{\dag}\mathbf{w}_{s}|^2|\mathbf{h}_{d}^{\dag}\mathbf{B}\mathbf{v}|}{N_0|\mathbf{h}_{d}^{\dag}\mathbf{B}\mathbf{v}|+
\epsilon\|\mathbf{B}\mathbf{v}\|
\sqrt{N_0P_s|\mathbf{w}_{r}^{\dag}\mathbf{H}_{r}^{\dag}\mathbf{w}_{s}|^2+N_0^2}},
\end{small}
\end{equation}
\begin{equation}
\begin{small}
\gamma_{e}=\frac{P_s|\mathbf{w}_{r}^{\dag}\mathbf{H}_{r}^{\dag}\mathbf{w}_{s}|^2
\epsilon\|\mathbf{B}\mathbf{v}\|}{N_0\epsilon\|\mathbf{B}\mathbf{v}\|+|\mathbf{h}_{d}^{\dag}\mathbf{B}\mathbf{v}|\sqrt{(N_0P_s\mathbf{w}_{r}^{\dag}\mathbf{H}_{r}^{\dag}\mathbf{w}_{s}|^2\!+\!N_0^2)}}.
\end{small}
\end{equation}

It is obvious that $\gamma_d$ is an increasing function of $|\mathbf{h}_{d}^{\dag}\mathbf{B}\mathbf{v}|/\epsilon\|\mathbf{B}\mathbf{v}\|$, while $\gamma_e$ is a decreasing function of that.  
Therefore, the WCSR maximization problem $\mathrm{P}2.3$ is equivalent to maximizing the ratio of $|\mathbf{h}_{d}^{\dag}\mathbf{B}\mathbf{v}|^2$ and $\epsilon^2\|\mathbf{B}\mathbf{v}\|^2$, i.e.,
\begin{equation*}
\begin{small}
(\mathrm{P}2.3.1):\max_{\mathbf{v}} \frac{\mathbf{v}^{\dag}\mathbf{B}^{\dag}\mathbf{h}_{d}\mathbf{h}_{d}^{\dag}\mathbf{B}\mathbf{v}}{ \epsilon^2\mathbf{v}^{\dag}\mathbf{B}^{\dag}\mathbf{B}\mathbf{v}}
\end{small}
\end{equation*}
\begin{equation}\label{eq_27}
\begin{small}
s.t.~\mathbf{v}^{\dag}\mathbf{B}^{\dag}\mathbf{B}\mathbf{v}=1.
\end{small}
\end{equation}
The objective function of (\ref{eq_27}) is a generalized eigenvector problem, and 
the solution of which is associated with the eigenvector of the matrix $\frac{1}{\epsilon^2}\mathbf{B}^{\dag}\mathbf{h}_{d}\mathbf{h}_{d}^{\dag}
\mathbf{B}$ \cite{LunDong2010, Jing2015}. Let $\mathbf{q}^{unit}_{1}$ be the unit-norm eigenvector of this matrix corresponding to its maximum eigenvalue, we have $\mathbf{v}^{\star}=\mu_1\mathbf{q}^{unit}_{1}$,  
where $\mu_1$ is a scalar determined by the power constraint, and is equal to
\begin{equation}
\mu_1=\sqrt{\frac{1}{(\mathbf{q}^{unit}_{1})^{\dag}\mathbf{B}^{\dag}\mathbf{B}\mathbf{q}^{unit}_{1}}}.
\end{equation}

Substituting $\mathbf{v}^{\star}$ into (\ref{de_pt}), if one of the constraints $C_1$, $C_2$ or $C_3$ is satisfied, 
the optimal transmit beamforming vector can be obtained 
\begin{equation}\label{eq_w1}
\mathbf{w}_{t}^{\star}=\sqrt{\frac{1}{(\mathbf{q}^{unit}_{1})^{\dag}\mathbf{B}^{\dag}\mathbf{B}\mathbf{q}^{unit}_{1}}}\mathbf{B}\mathbf{q}^{unit}_{1}.
\end{equation}
Otherwise, $\mathbf{w}_{t}^{\star}$ should be solved in the following part. 

\subsubsection{All harvested energy for information transmission}

For $\delta^{\star}=1$, the optimization problem of WCSR can be readily formulated as
\begin{equation*}
(\mathrm{P}2.3.2):\max_{\mathbf{v}}  \frac{\mathbf{v}^{\dag}\tilde{\mathbf{R}}_{fE}\mathbf{v}}{\mathbf{v}^{\dag}\tilde{\mathbf{R}}_{fD}\mathbf{v}}\cdot
\frac{\mathbf{v}^{\dag}\tilde{\mathbf{R}}_{RD}\mathbf{v}}{\mathbf{v}^{\dag}\tilde{\mathbf{R}}_{RE}\mathbf{v}}
\end{equation*}
\begin{equation}\label{eq_30}
s.t.~\mathbf{v}^{\dag}\mathbf{B}^{\dag}\mathbf{B}\mathbf{v}=1,
\end{equation}
where matrices $ \mathbf{R}_{RD}=\eta P_s^2|\mathbf{h}_{r1}^{\dag}\mathbf{w}_{H}|^2|\mathbf{w}_{r}^{\dag}\mathbf{H}_{r}^{\dag}\mathbf{w}_{s}|^2\mathbf{B}^{\dag}\mathbf{h}_{d}\mathbf{h}_{d}^{\dag}\mathbf{B}$, $\tilde{\mathbf{R}}_{fD}=(N_0P_s|\mathbf{w}_{r}^{\dag}\mathbf{H}_{r}^{\dag}\mathbf{w}_{s}|^2+N_0^2)(\mathbf{B}^{\dag}\mathbf{B}-\eta\mathbf{B}^{\dag}\mathbf{f}\mathbf{f}^{\dag}\mathbf{B})
+\eta P_sN_0|\mathbf{h}_{r1}^{\dag}\mathbf{w}_{H}|^2\mathbf{B}^{\dag}\mathbf{h}_{d}\mathbf{h}_{d}^{\dag}\mathbf{B}$, 
and $\tilde{\mathbf{R}}_{RD}=\mathbf{R}_{RD}+\tilde{\mathbf{R}}_{fD}$. 
Similarly, matrices $\mathbf{R}_{RE}=\eta P_s^2|\mathbf{h}_{r1}^{\dag}\mathbf{w}_{H}|^2|\mathbf{w}_{r}^{\dag}\mathbf{H}_{r}^{\dag}\mathbf{w}_{s}|^2\epsilon^2\mathbf{B}^{\dag}\mathbf{B}$, 
$\tilde{\mathbf{R}}_{fE}=(N_0P_s|\mathbf{w}_{r}^{\dag}\mathbf{H}_{r}^{\dag}\mathbf{w}_{s}|^2+N_0^2)(\mathbf{B}^{\dag}\mathbf{B}-\eta\mathbf{B}^{\dag}\mathbf{f}\mathbf{f}^{\dag}\mathbf{B})
+\eta P_sN_0|\mathbf{h}_{r1}^{\dag}\mathbf{w}_{H}|^2\epsilon^2\mathbf{B}^{\dag}\mathbf{B}$, 
and $\tilde{\mathbf{R}}_{RE}=\mathbf{R}_{RE}+\tilde{\mathbf{R}}_{fE}$.

Since both $\tilde{\mathbf{R}}_{fD}$ and $\tilde{\mathbf{R}}_{fE}$ are diagonal matrices, the objective function of $\mathrm{P}2.3.2$ is a product of two correlated generalized eigenvector problems, which is 
in general intractable \cite{LunDong2010}. 
To simplify the analysis and 
solve the problem, a suboptimal solution is proposed to maximize the upper and lower bounds of the objective function of (\ref{eq_30}). 

As it is known, the maximum value and the minimum value of the ratio $\mathbf{v}^{\dag}\tilde{\mathbf{R}}_{fE}\mathbf{v}/\mathbf{v}^{\dag}
\tilde{\mathbf{R}}_{fD}\mathbf{v}$
are corresponding to the maximum eigenvalue $\lambda_{max}$ and the minimum eigenvalue $\lambda_{min}$ of the matrix $\tilde{\mathbf{R}}_{fD}^{-1}\tilde{\mathbf{R}}_{fE}$, respectively. 
Therefore, the lower and upper bounds of the objective function can be written as
\begin{equation}\label{eq_31}
\begin{small}
\lambda_{min}
\frac{\mathbf{v}^{\dag}\tilde{\mathbf{R}}_{RD}\mathbf{v}}{\mathbf{v}^{\dag}\tilde{\mathbf{R}}_{RE}\mathbf{v}}\leq \frac{\mathbf{v}^{\dag}\tilde{\mathbf{R}}_{fE}\mathbf{v}}{\mathbf{v}^{\dag}\tilde{\mathbf{R}}_{fD}\mathbf{v}}\cdot
\frac{\mathbf{v}^{\dag}\tilde{\mathbf{R}}_{RD}\mathbf{v}}{\mathbf{v}^{\dag}\tilde{\mathbf{R}}_{RE}\mathbf{v}} \leq \lambda_{max}
\frac{\mathbf{v}^{\dag}\tilde{\mathbf{R}}_{RD}\mathbf{v}}{\mathbf{v}^{\dag}\tilde{\mathbf{R}}_{RE}\mathbf{v}},
\end{small}
\end{equation}
in which the analytical expressions of $\lambda_{max}$ and $\lambda_{min}$ are given on the top of next page.
\begin{figure*}[ht]
\normalsize
\setcounter{mytempeqncnt}{\value{equation}}
\setcounter{equation}{28}
\begin{equation}\label{eq_32}
\lambda_{max}=\max_{i}\left(\frac{\eta P_s|\mathbf{h}_{r1}^{\dag}\mathbf{w}_{H}|^2\epsilon^2+(P_s|\mathbf{w}_{r}^{\dag}\mathbf{H}_{r}^{\dag}\mathbf{w}_{s}|^2+N_0)(1-\eta|\mathbf{f}(i)|^2)}
{\eta P_s|\mathbf{h}_{r1}^{\dag}\mathbf{w}_{H}|^2|\mathbf{h}_{d}(i)|^2+(P_s|\mathbf{w}_{r}^{\dag}\mathbf{H}_{r}^{\dag}\mathbf{w}_{s}|^2+N_0)(1-\eta|\mathbf{f}(i)|^2)}\right),
\end{equation}
\setcounter{equation}{29}
\begin{equation}\label{eq_33}
\lambda_{min}=\min_{i}\left(\frac{\eta P_s|\mathbf{h}_{r1}^{\dag}\mathbf{w}_{H}|^2\epsilon^2+(P_s|\mathbf{w}_{r}^{\dag}\mathbf{H}_{r}^{\dag}\mathbf{w}_{s}|^2+N_0)(1-\eta|\mathbf{f}(i)|^2)}
{\eta P_s|\mathbf{h}_{r1}^{\dag}\mathbf{w}_{H}|^2|\mathbf{h}_{d}(i)|^2+(P_s|\mathbf{w}_{r}^{\dag}\mathbf{H}_{r}^{\dag}\mathbf{w}_{s}|^2+N_0)(1-\eta|\mathbf{f}(i)|^2)}\right).
\end{equation}
\setcounter{equation}{\value{mytempeqncnt}}
\hrulefill
\vspace*{4pt} 
\end{figure*}
\setcounter{equation}{30}

With the given $\lambda_{max}$ and $\lambda_{min}$, the solution of $\mathbf{v}$ maximizing the lower or upper bound in (\ref{eq_31}) is also the solution of the generalized eigenvector problem $\mathbf{v}^{\dag}\tilde{\mathbf{R}}_{RD}\mathbf{v}
/\mathbf{v}^{\dag}\tilde{\mathbf{R}}_{RE}\mathbf{v}$, i.e., 
$\mathbf{v}^{\star}=\mu_2\mathbf{q}_{2}^{unit}$,   
wherein $\mathbf{q}_{2}^{unit}$ is the unit-norm eigenvector of matrix $\tilde{\mathbf{R}}_{RE}^{-1}\tilde{\mathbf{R}}_{RD}$ corresponding to its maximum eigenvalue, and $\mu_2$ is a scalar chosen to satisfy the power constraint,  $\mu_2=\sqrt{1/[(\mathbf{q}^{unit}_{2})^{\dag}\mathbf{B}^{\dag}
\mathbf{B}\mathbf{q}^{unit}_{2}]}$. 
\begin{rem}
The above suboptimal solution works well when $\lambda_{\max}\approx\lambda_{\min}$. Based on (\ref{eq_32}) and (\ref{eq_33}), $\lambda_{\max}\approx\lambda_{\min}$ holds if the signal power at the relay is much greater than that at the destination and eavesdropper, i.e., $P_s|\mathbf{w}_{r}^{\dag}\mathbf{H}_{r}^{\dag}\mathbf{w}_{s}|^2 \gg \frac{\eta P_s|\mathbf{h}_{r1}^{\dag}\mathbf{w}_{H}|^2}{1-\eta|\mathbf{f}(i)|^2}|\mathbf{h}_{d}(i)|^2$ and $P_s|\mathbf{w}_{r}^{\dag}\mathbf{H}_{r}^{\dag}\mathbf{w}_{s}|^2 \gg \frac{\eta P_s|\mathbf{h}_{r1}^{\dag}\mathbf{w}_{H}|^2}{1-\eta|\mathbf{f}(i)|^2}\epsilon^2
$ for $i=1,2,\ldots, M$. 
Since in this case, the bounds in (\ref{eq_31}) are tight and the solution $\mathbf{w}_{t}^{\star}=\mu_{2}\mathbf{B}\mathbf{q}_{2}^{unit}$ that maximizes the bounds of WCSR is near-optimal. 
If $\lambda_{\max}\approx\lambda_{\min}$ does not hold,
the above method may not perform well and 
the other suboptimal method proposed in the following could be used instead.
\end{rem}

\begin{itemize}
\item 
Given an information leakage threshold ($\frac{1}{2}\log_2(1+\gamma_{ewc})\leq \gamma_0$), 
the information rate maximization problem at destination can be formulated as
\begin{equation*}
(\mathrm{P}2.3.3):\max_{\mathbf{v}}  \frac{\mathbf{v}^{\dag}\tilde{\mathbf{R}}_{RD}\mathbf{v}}{\mathbf{v}^{\dag}\tilde{\mathbf{R}}_{fD}\mathbf{v}}
\end{equation*}
\begin{equation}\label{eq_34}
s.t. \left\{\begin{array}{ll} 
&\mathbf{v}^{\dag}\mathbf{B}^{\dag}\mathbf{f}^{\dag}\mathbf{f}\mathbf{B}\mathbf{v} \leq g(\epsilon)\\
&\mathbf{v}^{\dag}\mathbf{B}^{\dag}\mathbf{B}\mathbf{v}=1\\
\end{array}\right..
\end{equation}
where 
\begin{equation}
\begin{small}
g(\epsilon)=\frac{1}{\eta}\!-\!\frac{P_s|\mathbf{h}_{r1}^{\dag}\mathbf{w}_{H}|^2\epsilon^2}{P_s|\mathbf{w}_{r}^{\dag}\mathbf{H}_{r}^{\dag}\mathbf{w}_{s}|^2\!+\!N_0}\left(\frac{P_s|\mathbf{w}_{r}^{\dag}\mathbf{H}_{r}^{\dag}\mathbf{w}_{s}|^2}{N_0(2^{2 \gamma_0}-1)}-\!1\!\right).
\end{small}
\end{equation}

To solve the above problem, an extra coefficient $\zeta \in [0,1]$ is introduced and used to form the inequality constraint in (\ref{eq_34}) as an equality constraint, $\mathbf{v}^{\dag}\mathbf{B}^{\dag}\mathbf{f}^{\dag}\mathbf{f}\mathbf{B}\mathbf{v} =\zeta g(\epsilon)$. 
By substituting the reformed inequality constraint and the other equality constraint into the objective function, the optimization problem $\mathrm{P}2.3.3$ can be rewritten as a generalized eigenvector problem. 
Then, the optimal $\mathbf{v}^{\star}=\mu_3\mathbf{q}^{unit}_{3}$ is obtained, 
where $\mathbf{q}^{unit}_{3}$ is the unit-norm eigenvector of the matrix $\hat{\mathbf{R}}_{fD}^{-1}\hat{\mathbf{R}}_{RD}$
corresponding to its maximum eigenvalue, in which  $\hat{\mathbf{R}}_{fD}=\eta P_sN_0|\mathbf{h}_{r1}^{\dag}\mathbf{w}_{H}|^2\mathbf{B}^{\dag}\mathbf{h}_{d}\mathbf{h}_{d}^{\dag}\mathbf{B}+
(N_0P_s|\mathbf{w}_{r}^{\dag}\mathbf{H}_{r}^{\dag}\mathbf{w}_{s}|^2+N_0^2)(1-\eta\zeta g(\epsilon))\mathbf{B}^{\dag}\mathbf{B}$ and $\hat{\mathbf{R}}_{RD}=\mathbf{R}_{RD}+\hat{\mathbf{R}}_{fD}$. 
$\mu_3$ in the solution is a scalar guaranteeing the power constraint, and $\mu_3=\sqrt{1/[(\mathbf{q}^{unit}_{3})^{\dag}\mathbf{B}^{\dag}
\mathbf{B}\mathbf{q}^{unit}_{3}]}$. 
In addition, the optimal value of the introduced coefficient $\zeta$ can be obtained  by bisection method. 
Finally, the optimal solution of transmit beamformer is derived to be $\mathbf{w}_{t}^{\star}=\mu_3\mathbf{B}\mathbf{q}^{unit}_{3}$. 
\end{itemize}

\subsection{WCSR Optimization Algorithm Description}

With the obtained transmit beamformer at relay, the power allocation ratio can be updated. 
Since $\mathbf{w}_{r}$ and $\mathbf{w}_{s}$ are only associated with the source-to-relay channel coefficients, 
they are proved to have no change in Proposition 3. 
Consequently, to maximize WCSR, we can first calculate the optimal information and receive beamformers, and then achieve the transmit beamformer and power allocation ratio iteratively.
The whole process of WCSR optimization is summarized in the following Algorithm 1.

\begin{algorithm}
  \caption{Worst-Case Secrecy Rate Optimization}
  \label{alg:optimal}
  \begin{algorithmic}[1]
  \STATE 
  Do SVD decomposition of the source-to-relay channel matrix $\mathbf{H}_{r}$, $\mathbf{H}_{r}^{\dag}=\mathbf{U}\mathbf{\Sigma}\mathbf{V}^{\dag}$.  
  \STATE Find out the largest singular value from  $\mathbf{\Sigma}$, and get its corresponding right-singular vector and left-singular vector $|\mathbf{u}^{unit}\|^2=1$, $\|\mathbf{v}^{unit}\|^2=1$.  
  \STATE Get the optimal $ \left(\mathbf{w}_{s}^{\star}, \mathbf{w}_{r}^{\star}\right)=\left(\mathbf{v}^{unit}, \mathbf{u}^{unit}\right)$.
 
  \STATE Calculate the transmit beamforming vector $\mathbf{w}_{t}^{\star}$ based on equation (\ref{eq_w1}).

  \STATE Check whether one of constraints $C_1$, $C_2$ and $C_3$ is satisfied with the obtained $\mathbf{w}_{t}^{\star}$.
  \IF{$C_1|C_2|C_3$}
  \STATE The power allocation ratio $\delta^{\star}=1/(\eta|\mathbf{f}^{\dag}\mathbf{w}_{t}^{\star}|^2+Q)$. 
  \ELSE 
  \STATE The power allocation ratio $\delta^{\star}=1$. $\mathbf{w}_{t}^{\star}$ is calculated by using methods proposed in Section III-C-2. 
  \ENDIF  
  \STATE Calculate $R_{wc}(\delta^{\star},\mathbf{w}_{s}^{\star}, \mathbf{w}_{r}^{\star},\mathbf{w}_{t}^{\star})$.
 \end{algorithmic}
\end{algorithm}

It should be noted that since $ (\delta^{\star},\mathbf{w}_{s}^{\star},\mathbf{w}_{r}^{\star},\mathbf{w}_{t}^{\star})$ of $\mathrm{P}1$ is obtained iteratively, 
the optimality of them depends on the optimality of the obtained transmit beamformer. If the transmit beamformer $\mathbf{w}_{t}^{\star}$ in (\ref{eq_w1}) can be achieved, the final solution is optimal for the original problem $\mathrm{P}1$. Otherwise, the obtained solution is suboptimal. 

\section{Time Switching Based Relaying}
\begin{figure}[h]
\centering
\includegraphics[width=0.4\textwidth]{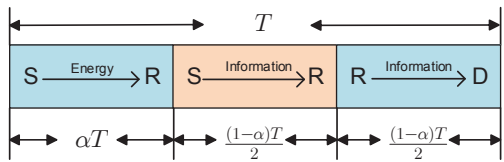}
\caption{Time switching based wireless powered relaying system model.}
\label{TS_HD_model}
\end{figure} 
TSR secure system \cite{Zeng2015WCL,Nasir2013TWC} is included as a benchmark for performance comparisons. The secure communication based on TSR protocol consists of three phases, as shown in \figref{TS_HD_model}. One phase for energy harvesting and the other two phases for information transmission from source to destination with the help of a HD relay. 
Using TSR, the received SINR at the destination is expressed as
\begin{equation}
\gamma_d^{TSR}=\frac{P_s P_r|\mathbf{h}_{d}^{\dag}\mathbf{w}_{t}|^2|\mathbf{h}_{r1}^{\dag}\mathbf{w}_{s}|^2}{P_r N_0|\mathbf{h}_{d}^{\dag}\mathbf{w}_{t}|^2+P_s N_0|\mathbf{h}_{r1}^{\dag}\mathbf{w}_{s}|^2+N_0^2}, 
\end{equation}
where $P_r$ denotes transmit power at the relay, $P_r=\delta \frac{2\alpha}{1-\alpha}\eta P_{s}|\mathbf{h}_{r1}^{\dag}\mathbf{w}_{H}|^2$, and $\delta$ is defined as the power allocation ratio, $0\leq \delta \leq 1$. 
Additionally, beamformers $\mathbf{w}_{s}$ and $\mathbf{w}_{H}$ are both designed based on MRT, i.e., $\mathbf{w}_{s}=\frac{\mathbf{h}_{r1}}{\|\mathbf{h}_{r1}\|}$ and $\mathbf{w}_{H}=\frac{\mathbf{h}_{r1}}{\|\mathbf{h}_{r1}\|}$.

The received SINR at the eavesdropper can be expressed as
\begin{equation}
\gamma_{ewc}^{TSR}=\frac{P_s P_r(|\bar{\mathbf{h}}_{e}^{\dag}\mathbf{w}_{t}|+\epsilon\|\mathbf{w}_{t}\|)^2|\mathbf{h}_{r1}^{\dag}\mathbf{w}_{s}|^2}{P_r N_0(|\bar{\mathbf{h}}_{e}^{\dag}\mathbf{w}_{t}|\!+\!\epsilon\|\mathbf{w}_{t}\|)^2\!+\!P_s N_0|\mathbf{h}_{r1}^{\dag}\mathbf{w}_{s}|^2\!+\!N_0^2}.\\
\end{equation}

Then, the WCSR for TSR secure systems is expressed as
\begin{equation}
R_{wc}=\frac{1-\alpha}{2}[\log_{2}(1+\gamma_d^{TSR})-\log_2(1+\gamma_{ewc}^{TSR})]^{+}.
\end{equation}
Finally, the WCSR maximization problem is formulated as
\begin{equation*}
(\mathrm{P}3):~\max_{\alpha,\delta,\mathbf{w}_{t}} R_{wc}(\alpha,\delta,\mathbf{w}_{t})
\end{equation*}
\begin{equation}
s.t. \left\{\begin{array}{ll} 
&P_{r}=2\delta P_{s}\frac{\eta \alpha}{1-\alpha} |\mathbf{h}_{r1}^{\dag}\mathbf{w}_{H}|^2\\
&\|\mathbf{w}_{t}\|^2= 1\\
&0< \alpha < 1\\
&0\leq \delta \leq 1\\ 
\end{array}\right..
\end{equation} 

It is obvious that the transmit beamformer $\mathbf{w}_{t}$ and the power allocation ratio $\delta$ as well as the TS coefficient $\alpha$ are coupled with each other, and the above problem is non-convex. 
To solve problem $\mathrm{P}3$ effectively, an iterative algorithm is summarized in Algorithm 2.
\begin{algorithm}
  \caption{Iterative Algorithm for WCSR Optimization}
  \label{alg:Alternate}
  \begin{algorithmic}[1]
  
  \STATE Initialize $\mathbf{w}_{s}=\frac{\mathbf{h}_{r1}}{\|\mathbf{h}_{r1}\|}$,  $\mathbf{w}_{H}=\frac{\mathbf{h}_{r1}}{\|\mathbf{h}_{r1}\|}$, 
  $\mathbf{w}_{t}^{0}$, $\alpha^{0}$, $\delta^{0}$, $\varepsilon$, 
calculate $R_{wc}(\mathbf{w}_{t}^{0},\alpha^{0},\delta^{0})$, and set $n=0$.

  \REPEAT  
  \STATE $n=n+1$
  
  \STATE For the fixed $\mathbf{w}_{t}^{n-1},\alpha^{n-1}$, find $\delta^{n}$ using \emph{Lemma 1} proposed in \cite{Jing2017TVT}.
  
  \STATE With $\delta^{n}$ and fixed $\alpha^{n-1}$, calculate $\mathbf{w}_{t}^{n}$ using an approach similar to Section III-C.
  
  \STATE For the obtained $\delta^{n}$ and $\mathbf{w}_{t}^{n}$, find the numerical solution of $\alpha^{n}$, and calculate the secrecy rate $R_{wc}^{n}=R_{wc}(\mathbf{w}_{t}^{n},\alpha^{n},\delta^{n})$.
  
  \UNTIL {$|R_{wc}^{n}-R_{wc}^{n-1}|\leq \varepsilon$}.
  
  \RETURN $(\mathbf{w}_{t}^{\star},\alpha^{\star},\delta^{\star})=
  (\mathbf{w}_{t}^{n},\alpha^{n},\delta^{n})$
  \end{algorithmic}
\end{algorithm}

\section{Simulation Results and Analysis}
In this section, to validate the accuracy and effectiveness of the theoretical results developed in previous sections and evaluate WCSR performance of the wireless-powered relay system with S-ER, Monte Carlo simulations are done. 
All the simulation results are obtained by averaging over $10^{4}$ independent trials.
Unless otherwise specified, the source has five antennas ($N=5$), while the wireless-powered relay is equipped with four antennas, in which one antenna is used for information reception and energy harvesting, the other three antennas ($M=3$) used for both information transmission and reception. 
Without loss of generality, 
the source-relay distance and the relay-destination distance are both normalized to unit value \cite{cai2014}, and the relay-eavesdropper distance $d_{re}$ is set to be $1.2 d_{rd}$. 
All channels are assumed to be Rayleigh fading channels, and the mean value of each channel coefficient is associated with the distance between two nodes, i.e., $\sigma_{s}^2=d_{sr}^{-m}$, $\sigma_{d}^2=d_{rd}^{-m}$ and $\sigma_{e}^2=d_{re}^{-m}$, in which the path loss exponent is set to be $m=3$.
Besides, the mean value of LI channel coefficient $\|\mathbf{f}\|^2$ is set to be $\lambda_{f}=0.2$, and the energy harvesting efficiency $\eta=0.8$.

\subsection{Uncertainty Region of ECSI}

\begin{figure}[t]
\centering
\includegraphics[width=0.5\textwidth]{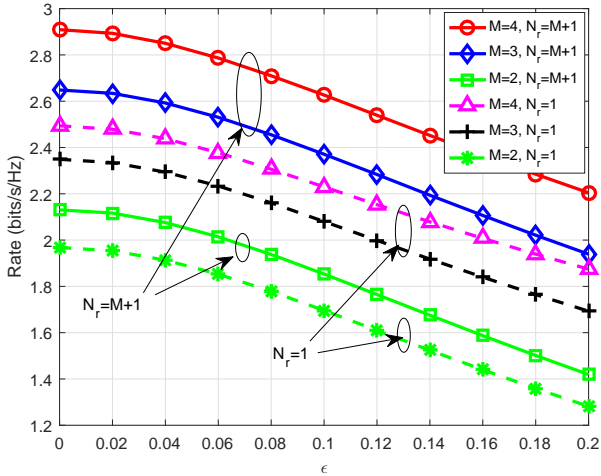}
\caption{WCSR versus the level of bounded uncertainty region of ECSI with different number of transmit antennas at the relay $M=[2,3,4]$. Additionally, 
the information leakage threshold 
is set to be $\gamma_0=10^{-6}$~bits/s/Hz.}
\label{fig_err}
\end{figure}
In a multi-antenna secure system, it is generally known that the secure beamformer design depends on the availability of CSI, especially the CSI of wiretap channels. 
If the source and the relay have only partial ECSI, there is a high probability in information leakage at the eavesdropper, finally leading to WCSR loss. 
Thus, in order to show how important it is to
take imperfect ECSI into account for secure beamformer design, 
the WCSR performance versus the uncertainty region of ECSI is simulated. Results are shown in \figref{fig_err} and  reveal that the increase of the level of the bounded eavesdropper channel uncertainty has destructive effect on the system performance. That is because this makes the secure beamformer design inefficient, leading to more information leakage to the eavesdropper, finally resulting in substantial WCSR degradation. 
In addition, results also show that
compared with the scheme exploiting only one receive antenna for information reception (i.e., $N_r=1$), the antenna reuse scheme (i.e., $N_r=M+1$) achieves better WCSR performance. 
Intuitively, this is because more spatial degrees of freedom offered by reusing transmit antenna elements facilitate the information reception, enhance the received signal strength at relay.
This results in an increase of received information rate at the destination, hence brings more WCSR gains. 
Moreover, putting more antennas at relay $M=[2,3,4]$ would also offer more spatial degrees of freedom for information retransmission, and consequently improve the system WCSR. 
This is rather intuitive, since more antennas result in sharper information beams, which in turn  yield higher secrecy performance.

\subsection{Power Allocation}

\begin{figure}[!t]
\centering
\includegraphics[width=0.5\textwidth]{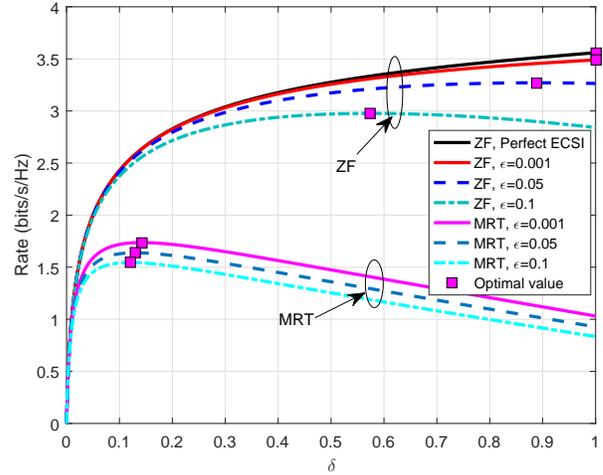}
\caption{WCSR performance versus the ratio between the actual transmit power and the harvested power at the relay, i.e., $\delta=P_r/P_H$, where $P_H$ is the harvested power at the relay, $P_H=\frac{E_H}{T/2}$.}
\label{fig_delta}
\end{figure}
As analyzed in Section III, due to the presence of the eavesdropper node, 
using the harvested energy fully for information transmission 
has both constructive and destructive effects on the system performance. 
The increase of transmit power at the relay 
helps information retransmission, 
and consequently improves the received SINR at the destination. 
Meanwhile, increasing the transmit power at relay also enhances the information strength at eavesdropper, resulting in more information leakage. 
Hence, there exists an optimal power allocation strategy at relay maximizing the WCSR performance. 
To illustrate the effects of power allocation strategy on secrecy performance and demonstrate the accuracy of the derived optimal power allocation ratio solution in Proposition 1, 
the WCSR performance with different power allocation ratio is simulated. 
In doing so, we focus on a single transmission block with the following randomly generated channel coefficients: 
the source-to-relay channel $\mathbf{h}_{r1}=[-0.9693 + 0.4571i,-1.4266 + 0.3548i,-1.8713 + 0.8418i, 0.7243 - 0.0702i,-0.9796 + 0.3818i]^{T}$, and 

\begin{small}
\begin{eqnarray*}
\!\mathbf{H}_{r2}\!\!=\!\!
\!\begin{pmatrix}\!
\!0.5023 \!+\! 0.9428i\!\! & \!\!1.0247 \!-\! 0.7866i\!\! & \!\!-0.2742 \!+\! 0.7717i\\
\!0.0555 \!-\! 0.5340i\!\! & \!\!-1.5941 \!-\! 0.4515i\!\! & \!\!-0.1950 \!-\! 0.1796i\\
\!-0.7962 \!+\! 0.9197i\!\! & \!\!-1.0882 \!-\! 0.2271i\!\! & \!\! 0.3783 \!-\! 0.5202i\\
\!1.0129 \!-\! 1.0110i\!\! & \!\!-1.1428 \!-\! 0.3521i\!\! & \!\!-0.0615 \!-\! 0.2866i\\
\!0.0170 \!+\! 1.3779i\!\! & \!\!-1.2486 \!-\! 0.0870i\!\! & \!\! 0.1937 \!-\! 1.1660i\\
\end{pmatrix}\!.\!
\end{eqnarray*}
\end{small}
It is also assumed that the LI channel 
$\mathbf{f}=[-0.6791 + 0.0424i,0.5303 + 0.1144i,-0.5517 - 0.0069i]^{T}$, while the relay-to-destination channel
$\mathbf{h}_d =[-0.5039 + 0.3520i,0.4230 - 1.1293i,0.6480 + 1.6376i]^{T}$, and the relay-to-eavesdropper channel 
$ \bar{\mathbf{h}}_e =[0.6417 - 0.3991i,0.1765 - 0.9396i,-0.5278 - 0.8778i]^{T}$. All the above channel coefficients are randomly generalized following their statistical properties.
With different levels of bounded eavesdropper channel uncertainty $\epsilon=[0.001,0.05,0.1]$, 
the WCSR performance versus the ratio between the actual transmit power and the harvested power at the relay is included in \figref{fig_delta}.
It is shown that the simulation results are highly coherent with the theoretical solutions obtained in Proposition 1, and there exists an optimal power allocation ratio maximizing WCSR. 
Moreover, results also verify that compared with MRT beamforming design method, i.e., $\mathbf{w}_{t}^{\star}=\frac{\mathbf{h}_{d}}{\|\mathbf{h}_{d}\|}$, the proposed robust ZF beamforming method can considerably reduce the effects caused by the channel uncertainty on the system secrecy, and offers higher WCSR gain.

\subsection{Key Parameters}

\begin{figure}[!t]
\centering
\includegraphics[width=0.5\textwidth]{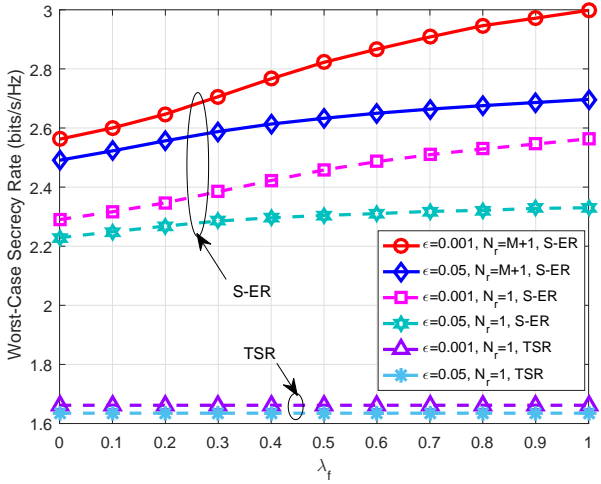}
\caption{WCSR performance versus the strength of LI for the proposed S-ER based secure system and traditional TSR secure system. The eavesdropper channel uncertainty region is set to be $\epsilon=[0.001,0.05]$.}
\label{fig_LI}
\end{figure}
In a S-ER based wireless-powered secure relay system, the strength of LI channel determines how much amount of interference energy the relay can recycle and reuse for information retransmission. 
Therefore, it is a key parameter determining the superiority of the secrecy performance. 
\figref{fig_LI} shows the WCSR performance of both S-ER based secure system and TSR system when the strength of LI channel is different. 
Contrary to traditional FD relaying systems, the LI does not interfere information transmission, but brings another energy resource at relay to support information relaying. 
Therefore, as shown in \figref{fig_LI}, increasing $\lambda_f$ leads to an increase of WCSR performance of the proposed S-ER based secure systems.  
While, the secrecy rate remains a constant because there is no interference in TSR systems. 
Since a novel energy resource (LI energy) is introduced by the FD operation at the relay in S-ER based system, its secrecy performance is significantly boosted compared with TSR systems. 
Even setting $\lambda_f=0$, in S-ER based systems, more WCSR gain can be obtained through simultaneous information transmission and energy harvesting strategy than in TSR systems. 
Specifically, with fixed eavesdropper channel uncertainty $\epsilon=0.001$, the proposed S-ER based secure system achieves more than $38\%-80\%$ WCSR gain compared with TSR system for both $N_r=M+1$ and $N_r=1$. 
With more higher eavesdropper channel uncertainty $\epsilon=0.05$, the proposed S-ER based secure system can also offer $36\%-65\%$ rate gain. 
Moreover, since LI brings extra energy resource at relay instead of harmful interference for information transmission, the relay does not have to spend extra energy or other communication resources to suppress it. This makes it possible for S-ER based secure systems to benefit from significant secrecy performance improvement without no extra resource consumptions.

\begin{figure}[!t]
\centering
\includegraphics[width=0.5\textwidth]{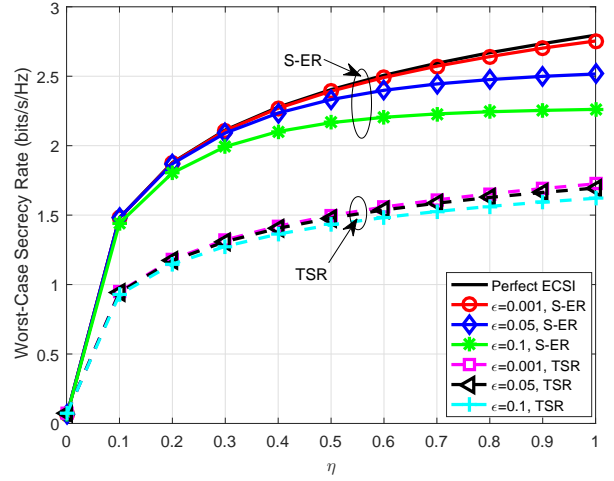}
\caption{WCSR performance versus energy harvesting efficiency for both S-ER based secure systems and TSR secure systems. The strength of LI channel is set to be $\lambda_f=0.2$.}
\label{fig_eta}
\end{figure}
As illustrated in \figref{fig_LI}, the strength of LI channel and the energy conversion efficiency factor $\eta$ determine how much energy  the relay can actually harvest from energy signals emitted by the source node and interference signals caused by LI channel.
Higher energy conversion efficiency allows wireless-powered relay to harvest more energy from energy signals, which in turn enhances the received information rate at the destination, finally improves the WCSR performance substantially, as shown in \figref{fig_eta}. 
Nevertheless, as the harvested energy increased along with $\eta$, it also leads to more serious information leakage, eventually results in the secrecy performance increasing slowly in high $\eta$ regime. 
Additionally, simulation results also show that the WCSR gap between the proposed S-ER based systems and TSR 
systems is more than $0.7$~bits/s/Hz in high $\eta$ regime, which is about 40\% rate gain. 
It indicates that with the introduced LI as extra energy resource for the wireless-powered relay node,  
the S-ER based system achieves much better WCSR performance than TSR systems.

\begin{figure}[!t]
\centering
\includegraphics[width=0.5\textwidth]{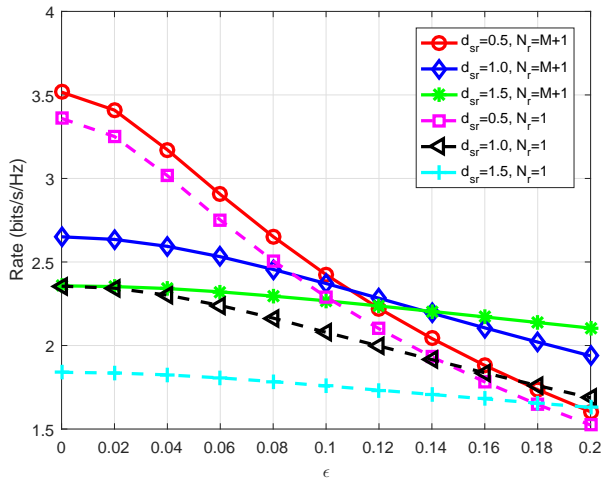}
\caption{WCSR performance for different source-relay distance $d_{sr}=[0.5,1,1.5]$. Without loss of generality, the strength of LI channel is set to be $\lambda_f=0.2$, and the energy harvesting efficiency is set to be $\eta=0.8$. }\label{fig_dis}
\end{figure}
The location of relay can also impact on the received information rate at receivers (including both destination and eavesdropper nodes). Therefore, the relay layout is of great importance in the cooperative wireless-powered systems \cite{cai2014,Jing2017TVT}. 
In this part, the WCSR performance for different relay locations is also simulated and investigated, and results are included in \figref{fig_dis}.  
Results show that with the eavesdropper channel uncertainty, the wireless-powered relay located near to the source or the destination offers optimal secrecy performance. 
The reason is that the relay node is solely powered by the source. 
In low $\epsilon$ regime, when the relay is located near to the source, it can harvest more energy for information relaying and hence improves WCSR performance. 
While in high $\epsilon$ regime, excessive amount of energy harvested at relay will lead to serious information leakage, and results in high secrecy performance loss. 
Therefore, in this case, the relay located near the source does not bring more performance gain but serious performance loss. 
On the contrary, the wireless-powered relay approaching to the destination can obtain better WCSR performance. Since 
when the relay node located near the destination, the required relay power to achieve the information transmission from relay to destination is also reduced, the destination can still achieve better performance.  

\section{Conclusion}
This work investigated robust secure beamforming design and relay power allocation of a wireless-powered relay system with imperfect ECSI. 
For efficient energy transfer and information relaying, a two-phase FD relaying protocol was investigated, 
in which the S-ER enabled relay recycles LI energy and utilizes the gathered energy for information relaying.  
Adopting the worst-case deterministic ECSI model, the WCSR optimization problem was formulated by jointly designing source and relay beamformers coupled with power allocation ratio. 
As the formulated problem is non-convex, an alternative method has been proposed to solve it more efficiently. 
Through converting the original optimization problem into three subproblems and solving them iteratively, 
the closed form solutions of the optimal source information beamformer, receive beamformer and transmit beamformer of the relay as well as optimal power allocation ratio were derived theoretically. 
Additionally, for comparison, the WCSR performance of TSR system has been included and analyzed. Intensive simulations have been done, and results have validated the effectiveness of the proposed iterative method and the accuracy of the derived variables.
Results also have revealed that the proposed S-ER based secure system outperforms TSR system significantly.
Moreover, it has been demonstrated that since the proposed antenna reuse scheme can capture more information, it boosts the WCSR performance substantially when compared with schemes  exploiting only one receive antenna.

\appendices

\section{Proof of Proposition 1}
According to the definition of WCSR in (\ref{eq_7}),
$R_{wc}>0$ is equivalent to  $\gamma_{d}>\gamma_{ewc}$, i.e., the received SINR at the destination is larger than that at the eavesdropper. 

Based on the expressions of $\gamma_{d}$ and $\gamma_{ewc}$ in equation (\ref{eq_8}) and (\ref{eq_12}), it is easy to know that both of them can be formulated in the form of $g(y)=c_{1}y/(c_2y+c_3)$, where $y$ denotes the channel condition, $|\mathbf{h}_{d}^{\dag}\mathbf{w}_{t}|^2$ or $(|\bar{\mathbf{h}}_{e}^{\dag}\mathbf{w}_{t}|+\epsilon\|\mathbf{w}_{t}\|)^2$.
It can be proved that $g(y)$ is a monotonically increasing function of $y$.
Consequently, if and only if $|\mathbf{h}_{d}^{\dag}\mathbf{w}_{t}|^2$ is larger than $(|\bar{\mathbf{h}}_{e}^{\dag}\mathbf{w}_{t}|+\epsilon\|\mathbf{w}_{t}\|)^2$, the received SINR at the destination is larger than the received SINR at the eavesdropper. The positive secrecy rate is achieved and the security of information transmission is guaranteed.
 Otherwise,  the secrecy rate will be non-positive.


\section{Proof of Proposition 2}
The secrecy rate is the function of power allocation ratio $\delta$. For analysis simplicity, introducing $x, 0\leq x \leq 1$, to substitute $\delta$, the rate can be formulated as a function $f(x)$ 
\begin{equation}
f(x)=\frac{1}{2}\left(\log_2(1+\frac{a_1x}{a_2x+1})-\log_2(1+\frac{a_3x}{a_4x+1})\right),
\end{equation}
where parameters $a_1$, $a_2$, $a_3$ and $a_4$ are expressed as
\begin{equation*}
a_1=\frac{\eta P_s^2|\mathbf{h}_{d}^{\dag}\mathbf{w}_{t}|^2|\mathbf{w}_{r}^{\dag}\mathbf{H}_{r}^{\dag}\mathbf{w}_{s}|^2|\mathbf{h}_{r1}^{\dag}\mathbf{w}_{H}|^2}{N_0P_s|\mathbf{w}_{r}^{\dag}\mathbf{H}_{r}^{\dag}\mathbf{w}_{s}^2+N_0^2}, \\
\end{equation*}
\begin{equation*}
a_2=\frac{\eta P_sN_0|\mathbf{h}_{d}^{\dag}\mathbf{w}_{t}|^2|\mathbf{h}_{r1}^{\dag}\mathbf{w}_{H}|^2}{N_0P_s|\mathbf{w}_{r}^{\dag}\mathbf{H}_{r}^{\dag}\mathbf{w}_{s}|^2+N_0^2}-\eta|\mathbf{f}^{\dag}\mathbf{w}_{t}|^2, \\
\end{equation*}
\begin{equation*}
a_3=\frac{\eta P_s^2(|\bar{\mathbf{h}}_{e}^{\dag}\mathbf{w}_{t}|+\epsilon\|\mathbf{w}_{t}\|)^2|\mathbf{w}_{r}^{\dag}\mathbf{H}_{r}^{\dag}\mathbf{w}_{s}|^2|\mathbf{h}_{r1}^{\dag}\mathbf{w}_{H}|^2}{N_0P_s|\mathbf{w}_{r}^{\dag}\mathbf{H}_{r}^{\dag}\mathbf{w}_{s}|^2+N_0^2}, \\
\end{equation*}
\begin{equation*}
a_4=\frac{\eta P_sN_0(|\bar{\mathbf{h}}_{e}^{\dag}\mathbf{w}_{t}|+\epsilon\|\mathbf{w}_{t}\|)^2|\mathbf{h}_{r1}^{\dag}\mathbf{w}_{H}|^2}{N_0P_s|\mathbf{w}_{r}^{\dag}\mathbf{H}_{r}^{\dag}\mathbf{w}_{s}|^2+N_0^2}-\eta|\mathbf{f}^{\dag}\mathbf{w}_{t}|^2. \\
\end{equation*}
It is obvious that $a_1>0$, $a_3>0$. Moreover, under the positive secrecy rate constraint proved in Proposition \ref{prop_1}, we have $a_1>a_3$, $a_2>a_4$.

The derivative of $f(x)$ is expressed as
\begin{eqnarray*}
\begin{small}
\begin{split}
\frac{df(x)}{dx}=&\frac{1}{\ln2} \frac{\left((a_1a_4^2-a_2^2a_3)+a_1a_3(a_4-a_2)\right)x^2}{(1\!+\!a_1x\!+\!a_2x)(1\!+\!a_2x)(1\!+\!a_3x\!+\!a_4x)(1\!+\!a_4x)}\\
&\!+\!\frac{1}{\ln2}\frac{2(a_1a_4\!-\!a_2a_3)x+a_1-a_3}{(1\!+\!a_1x\!+\!a_2x)(1\!+\!a_2x)(1\!+\!a_3x\!+\!a_4x)(1\!+\!a_4x)}.
\end{split}
\end{small}
\end{eqnarray*}
As the denominator of $df(x)/dx$ is always positive for $x\geq 0$, the sign of the numerator polynomial of $df(x)/dx$
\begin{eqnarray}
\begin{split}
h(x)=&\left((a_1a_4^2-a_2^2a_3)+a_1a_3(a_4-a_2)\right)x^2\\
&+2(a_1a_4-a_2a_3)x+a_1-a_3,\\
\end{split}
\end{eqnarray}
determines the monotonicity and optimality of $f(x)$. 
Defining $A=(a_1a_4^2-a_2^2a_3)+a_1a_3(a_4-a_2)$, $B=2(a_1a_4-a_2a_3)$, $C=a_1-a_3$, we have $B<0$, $C>0$, $h(0)=C>0$ and $\Delta=4a_1a_3(a_2-a_4)(a_1+a_2-a_3-a_4)>0$. It means that there are two roots for $h(x)=0$. In the following, the sign of $h(x)$ is analyzed.

1) For $A>0$, i.e., $\eta|\mathbf{f}^{\dag}\mathbf{w}_{t}|^2>Q$. Based on Vieta theory,  it is easy to prove that there are two positive roots for $h(x)=0$, and one of them is larger than $1$. Therefore, two cases of the sign of $h(1)=\left(1-\eta|\mathbf{f}^{\dag}\mathbf{w}_{t}|^2\right)^2-Q^2$ should be analyzed.

Case i) For $h(1)\geq 0$, we have $Q<\eta|\mathbf{f}^{\dag}\mathbf{w}_{t}|^2\leq1-Q$ and $Q<\frac{1}{2}$. In this case, both two roots are no less than $x=1$. So, the sign of $h(x)$ is always positive in the interval $0\leq x\leq 1$. That is, the function $f(x)$ is monotonically increasing with respect to $0\leq x\leq 1$, and the optimal solution is $x^{\star}=1$.

Case ii) For $h(1)< 0$, we have $\eta|\mathbf{f}^{\dag}\mathbf{w}_{t}|^2>\max(Q,1-Q)$ and $0<Q<1$. In this case, one root is larger than $x=1$ and the other one is less than $x=1$, i.e., the sign of $h(x)$ is changed from positive to negative. Thus $f(x)$ is concave, and the optimal solution equals to
 \begin{equation}\label{x_opt}
 x^{\star}=\frac{1}{\eta|\mathbf{f}^{\dag}\mathbf{w}_{t}|^2+Q}
 \end{equation}
 
2) For $A\leq 0$, i.e., $\eta|\mathbf{f}^{\dag}\mathbf{w}_{t}|^2\leq Q$. Based on Vieta theory,  it is easy to prove that there are one positive root and one negative root for $h(x)=0$, and the absolute value of the negative root is larger than the positive one, thus  $h(x)$ is monotonically decreasing in the interval $0\leq x\leq 1$.

Case i) $h(1)\geq 0$, we have  $\eta|\mathbf{f}^{\dag}\mathbf{w}_{t}|^2\leq \min(Q,1-Q)$ and $0<Q<1$. 
 Since $h(0)>0$, $h(1)\geq 0$ and $h(x)$ is monotonically decreasing, then we have $h(x)\geq 0$. Therefore, the function $f(x)$ is monotonically increasing for $0\leq x\leq 1$, and the optimal solution is $x^{\star}=1$.
 
Case ii) For $h(1)=0$, we have $\eta|\mathbf{f}^{\dag}\mathbf{w}_{t}|^2=1-Q$ and $\frac{1}{2}<Q<1$. In this case, with the condition $h(0)>0$, it is easy to prove that $h(x)\geq 0$ for $0\leq x\leq 1$, i.e., the function $f(x)$ is monotonically increasing in $0\leq x\leq 1$, and the optimal solution is $x^{\star}=1$.

 Case iii) For $h(1)< 0$, when $A<0$, we have $1-Q<\eta|\mathbf{f}^{\dag}\mathbf{w}_{t}|^2<\min(1,Q)$ and $Q>\frac{1}{2}$, and when $A=0$, we have $\eta|\mathbf{f}^{\dag}\mathbf{w}_{t}|^2=Q$ and $\frac{1}{2}<Q<1$. 
Specially, when $A<0$, the positive root is less than $x=1$,  i.e., the sign of $h(x)$ is changed from positive to negative. That is to say $f(x)$ is concave, and the optimal solution is the same as in equation (\ref{x_opt}).  
Alternatively, when $A=0$, $h(x)$ has only one positive root and less than $x=1$. Thus, similarly with $A<0$ case, $f(x)$ is also concave, and the optimal solution is (\ref{x_opt}). 
 
Consequently, combining case ii) in part 1) and case iii) in part 2), for constraints 
$C_1=\{0<Q<1,\eta|\mathbf{f}^{\dag}\mathbf{w}_{t}|^2>\max(Q,1-Q)\}$, or 
$C_2=\{Q>\frac{1}{2},1-Q<\eta|\mathbf{f}^{\dag}\mathbf{w}_{t}|^2<\min(1,Q)\}$, or
$C_3=\{\frac{1}{2}<Q<1, \eta|\mathbf{f}^{\dag}\mathbf{w}_{t}|^2=Q\}$, 
the function $f(x)$ is concave, and the optimal solution is derived as (\ref{x_opt}). Otherwise, the function $f(x)$ is monotonically increasing in $0\leq x\leq 1$, and the optimal solution is $x^{\star}=1$.

\section{Proof of Proposition 3}
As the source information beamformer and the relay receive beamformer always appear together with the source-to-relay channel coefficients in a form of production, the WCSR performance with respect to $|\mathbf{w}_{r}^{\dag}\mathbf{H}_{r}^{\dag}\mathbf{w}_{s}|^2$ is  investigated.  
Since the solution of the optimal power allocation ratio $\delta^{\star}$ in a piecewise expression, two cases should be analyzed.

1) For $\delta^{\star}=\frac{1}{\eta|\mathbf{f}^{\dag}\mathbf{w}_{t}|^2+Q}$, the received SINRs at the receive nodes (destination and eavesdropper) are expressed as
\begin{equation*}
\begin{small}
\gamma_{d}=\frac{P_s|\mathbf{w}_{r}^{\dag}\mathbf{H}_{r}^{\dag}\mathbf{w}_{s}|^2|\mathbf{h}_{d}^{\dag}\mathbf{w}_{t}|}{N_0|\mathbf{h}_{d}^{\dag}\mathbf{w}_{t}|\!\!+\!\!(|\bar{\mathbf{h}}_{e}^{\dag}\mathbf{w}_{t}|\!\!+\!\!\epsilon\|\mathbf{w}_{t}\|)\sqrt{N_0P_s|\mathbf{w}_{r}^{\dag}\mathbf{H}_{r}^{\dag}\mathbf{w}_{s}|^2\!+\!N_0^2}},
\end{small}
\end{equation*}
\begin{equation*}
\begin{small}
\gamma_{ewc}=\frac{P_s|\mathbf{w}_{r}^{\dag}\mathbf{H}_{r}^{\dag}\mathbf{w}_{s}|^2(|\bar{\mathbf{h}}_{e}^{\dag}\mathbf{w}_{t}|+\epsilon\|\mathbf{w}_{t}\|)}{N_0(|\bar{\mathbf{h}}_{e}^{\dag}\mathbf{w}_{t}|\!\!+\!\!\epsilon\|\mathbf{w}_{t}\|)\!\!+\!\!|\mathbf{h}_{d}^{\dag}\mathbf{w}_{t}|\sqrt{\!N_0P_s|\mathbf{w}_{r}^{\dag}\mathbf{H}_{r}^{\dag}\mathbf{w}_{s}|^2\!\!+\!\!N_0^2}}.
\end{small}
\end{equation*}
With the definition of $x=|\mathbf{w}_{r}^{\dag}\mathbf{H}_{r}^{\dag}\mathbf{w}_{s}|^2$, and $g(x)=\sqrt{P_s x+N_0}$, the WCSR can be formed as a function of $x$
\begin{equation}
f(x)=\log_2\left(\frac{a_1g(x)+ca_2}{a_2g(x)+ca_1}\right), 
\end{equation}
where $c=\sqrt{N_0}$, $a_1=|\mathbf{h}_{d}^{\dag}\mathbf{w}_{t}|$, $a_2=|\bar{\mathbf{h}}_{e}^{\dag}\mathbf{w}_{t}|+\epsilon\|\mathbf{w}_{t}\|$. 
The derivative function is expressed as
\begin{equation}\label{eq_55}
\frac{df(x)}{dx}=\frac{1}{\ln2}\frac{ca_2(a_1-a_2)}{(a_1g(x)+ca_2)(a_2g(x)+ca_1)}\frac{dg(x)}{dx}.
\end{equation}

As the denominator of (\ref{eq_55}) is always positive for $x\geq0$, the sign of $df(x)/dx$ is determined by the sign of 
its numerator. It is obvious that $g(x)$ is an increasing function of $x$, then we have the sign of its derivative function is positive, $dg(x)/dx>0$. Additionally, under the positive WCSR constraint, $a_1>a_2$ always holds. 
Therefore, $df(x)/dx$ is always greater than zero, and $f(x)$ is a monotonically increasing function of $x\geq0$.

2) For $\delta^{\star}=1$, the WCSR can be formed as the following function
\begin{equation}
h(x)=\frac{1}{2}\log_2\left(\frac{b_1x}{b_2x+1}\right)-\frac{1}{2}\log_2\left(\frac{b_3x}{b_4x+1}\right). 
\end{equation}
On the basis of Lemma 1 proposed in \cite{Jing2017TVT}, it is easily proved that $h(x)$ is also an increasing function of $x$.

Consequently, for any value of power allocation ratio $\delta^{\star}$, the WCSR is a monotonically increasing function of 
$|\mathbf{w}_{r}^{\dag}\mathbf{H}_{r}^{\dag}\mathbf{w}_{s}|^2$, and only if the maximum value of $|\mathbf{w}_{r}^{\dag}\mathbf{H}_{r}^{\dag}\mathbf{w}_{s}|^2$ is achieved, the WCSR performance can be maximized.




\end{document}